\documentclass[aps,prl,preprint,groupedaddress,showpacs]{revtex4-1}

\usepackage{amsmath}
\usepackage{amssymb}
\usepackage{graphicx} 
\usepackage{amsthm}

\newcommand{\R}{\mathbb R}

\newcommand{\Z}{\mathbb Z}
\newcommand{\X}{\mathbb{X}}

\newtheorem{theorem}{Theorem}
\newtheorem{definition}{Definition}
\newtheorem{lema}{Lemma}

\begin{document}


\title{Modelling Opinion Dynamics:  Theoretical analysis and continuous approximation}


\author{Juan Pablo Pinasco}
\affiliation{Departamento de Matem\'atica, Facultad de Ciencias Exactas y Naturales, Universidad de Buenos Aires  and  IMAS UBA-CONICET, Buenos Aires, Argentina.}
\author{Viktoriya Semeshenko}
\affiliation{Instituto Interdisciplinario de Econom\'ia Pol\'itica (IIEP-BAIRES), UBA,
CONICET, FCE, CABA, Argentina.}
\author{Pablo Balenzuela}
\affiliation{Departamento de F\'isica, Facultad de Ciencias Exactas y Naturales, Universidad de Buenos Aires and IFIBA, CONICET.
Buenos Aires, Argentina.}

\date{\today}

\begin{abstract}

Frequently we revise our first opinions after talking over with other individuals because we
get convinced. Argumentation is a verbal and social process aimed at convincing. It includes
conversation and persuasion.  In this case, the agreement is reached because the new arguments are
incorporated. In this paper we deal with a simple model of opinion formation with such persuasion dynamics, and  we find the exact analytical solutions for both, long and short range interactions. A novel theoretical approach has been used in order to solve the master equations of the model with non-local kernels.
Simulation results demonstrate an excellent agreement with results obtained by the theoretical estimation.

\end{abstract}


\maketitle


In group discussions individuals exchange arguments over a specific subject of conversation, and then selectively either incorporate what they have discovered or at least learn to understand one another better. That is to say, individuals may want to change their own opinions about an issue in order to get closer to or farther from others in the group. These human interactions give rise to the formation of different kinds of opinions in a society. At the end of the discussion the group will be characterized either by a so called opinion consensus or coexistence of opinions.

There are many approaches for modelling the dynamics of human opinions, and they all differ in their focus and complexity (\cite{Voter1,Voter2,WD2000,Masetal2013,LaRoccaetal2014}). G. Deffuant et al. \cite{WD2000} proposed a model of continuos opinions (D-W) based mainly on two ingredients: compromise and bounded confidence ($d$). The implicit conjecture is that individual opinion's shifts are due to interpersonal comparison process. Given a population of $N$ agents with initial opinions randomly equally distributed in some finite interval, the authors derived the analytical solution in the limit of local interactions ($d \rightarrow 0$). The compromise hypothesis assumes people tend to agree when they interact. Under the bounded confidence it is possible only in the case when their opinions are close enough. The last assumption has some mathematical reason beneath:  the master equation could be derived for the case of local interactions only (i.e., when two opinions are close enough, $d \rightarrow 0$). In
the present paper we will show that it can be extended even for  of long-range interactions.

The D-W model explains group induced shifts in individual choice in terms of interpersonal comparison process. By comparing himself with others in a group a member finds out that his position is uncomfortably discrepant, e.g., he is overly cautious or overly risky. Knowledge of this discrepancy presumably is necessary and sufficient to induce him to change his initial choice.
Another class of theories holds that merely knowing one is different from others is unimportant. Shifts in choice occur because during discussion a member is exposed to persuasive arguments which prior to discussion were not available to him. This line of reasoning goes back to Thorndike \cite{Thorndike}, and is reconsidered next in \cite{Nordhoy, Jean, Stoner, Vinokur, Burnstein}.

In the present paper, we take for a basis of study the D-W  model and address two questions: what will be the master equation that governs the evolution of opinions when the consensus emerges as a result of an  exchange of arguments. Second, how this master equation will change if we relax the hypothesis of bounded confidence and, as a consequence, long range interactions are included.

In order to answer these questions, we consider a model in which, whenever two
agents interact, their opinions will change by a discrete amount. In
the continuous limit, for a small bounded confidence interval,  the distribution of agent's
opinions is governed by a Porous Media equation backward in time as in \cite{WD2000},
with a slightly different time scaling. Let us mention that this equation cannot be solved numerically
due to high amplifications of numerical errors, given that it is not well-posed\cite{Levine74}.

If we remove the bounded confidence hypothesis, we cannot longer take the continuous limit as
before, and a new equation appears. It is possible to solve the equation explicitly using a method developed by Li and
Toscani \cite{Toscani}, which permits to find the exact solutions of the continuous
approximation of the master equations that govern the evolution of the system, and then compare
them with numerical simulations. We get a non-local Porous Media equation (different from
the one obtained in \cite{Toscani}, which can be thought also as a first
order hyperbolic equation with a nonlinear non-local flux). We will show that the distribution
of agent's opinions converges to a Dirac's delta function concentrated at the mean opinion of
the initial distribution. Let us observe that this partial differential equation develops a shock
at the median of the distribution, and the median value moves toward the mean.

Let us stress  that in the model of Deffuant et al \cite{WD2000} the step of the dynamics is a weighted average of the agent's opinions, and thus,
it is always proportional to the difference between the two opinions. However, in the present case it is always of the same order, and
 each agent moves his opinion by a fixed amount towards the other agent's opinion. Although
 the  equations obtained are similar, we can further relax the bounded confidence hypothesis obtaining a better description of social
 processes.

The paper is organized as follows. In the first section we derive the analytical solutions for the bounded confidence dynamics, and we obtain
an equation similar to the one of Deffuant et al \cite{WD2000}.
In the second section, we show that it is possible to derive the master equation when the hypothesis of  bounded confidence is relaxed
and long range interactions are included. In the last section we discuss the results and conclude.

\section{Dynamics with Bounded Confidence}

The D-W model can be formulated in a way that the arguments exchange process is included. In order to fulfil this objective, we propose an agent-based model, where  $\sigma(i)$ ($-1 \le \sigma \le 1$) stands for the opinion of agent $i$ about a certain subject.  We assume that every time two agents interact, they increase or decrease their opinions by a fixed quantity $h$  if these differ by less than a fixed quantity $d$ (bounded confidence).  We also retain  the compromise hypothesis by which, if agents $i$ and $j$ interact, and $\sigma(i)<\sigma(j)$,

\begin{eqnarray}
\sigma^*(i) &=& \sigma(i)+h, \nonumber \\
\sigma^*(j) &=& \sigma(j)-h,
\label{EqDynMod1P}
\end{eqnarray}
 that is, agents tend to agree with one another.  In this way, the consensus dynamics is not instantaneous and could be understood as a discussion process in which agents become closer with time.

 In order to obtain the master equations of this model, let us subdivide $[-1,1]$ in $M$ intervals $\{I_j\}_{1\le j\le M}$, of length $h$, and define:
 \begin{equation}
  s(j,t) =\frac{\#\{i : \sigma(i,t)\in I_j\}}{N},
 \end{equation}
for $1\le j \le M$,
as the density of agents with opinion $\sigma$ in the intervals $I_j$. We can extend $s$ to $\Z$ by zero outside $[1,M]$.
Let  $d=kh$ be fixed, and then an agent in $I_j$ can interact with other agents located in the intervals $I_{j-k}$, $I_{j-k+1}$,... $I_{j+k}$.

Let us deduce the master equation for  the density $s$.
Fixing some characteristic time $\tau$ related to the rate of interactions, we have
$$s(j, t+\tau)= s(j,t) + \frac{2}{N} (G(j,t) - L(j,t))$$
where  $G(j,t)$ stands for a gain term and $L(j,t)$ for a loss term. In time $\tau$ only a pair of agents move,
and then the proportion of agents $s_j$ can increase or decrease by $1/N$.  The factor 2 appears since we can choose an
agent located at $I_j$ as the first or the second agent in the interaction. The gain term $G$ is
computed as the probability of an interaction between some agent located at $I_{j+1}$ (respectively, $I_{j-1}$)
at time $t$ and another agent located at $I_i$ for $i\in [j+1-k, j]$ (resp., $i\in [j, j-1+k])$. The loss term $L$  is
computed as the probability of an interaction between some agent located at $I_j$ and another agent within its range of interaction, the intervals
$I_i$ with
$i \in [j-k, j+k]$, except at $I_j$, since in this case there are no changes. Therefore,
\begin{align}
\frac{N}{2}\Big( s(j,t+\tau) - s(j,t)\Big)  =& (G(j,t) - L(j,t)   )  \nonumber \\
 = & \Big( s(j+1)\sum_{i=1}^k s(j+1-i) +s(j-1)\sum_{i=1}^k s(j-1+i) \Big)  \nonumber \\ & - \Big(s(j)\sum_{i=1}^k [s(j+i)+s(j-i)] \Big),
\label{EqDiscr1}
\end{align}
where, for notation brevity, we have omitted the variable $t$ in the right hand side. If we re-arrange the terms, it reads as

\begin{align}
\frac{\tau N}{2}\, \frac{ s(j,t+\tau) - s(j,t)}{\tau} = & [ s(j+1)-s(j)]\sum_{i=1}^k s(j+1-i)  \nonumber \\
& +
[s(j-1)-s(j)]\sum_{i=1}^k s(j-1+i)  \nonumber\\
&- s(j)s(j+k)
+2s(j)s(j) - s(j)s(j-k).
\label{EqDiscr2}
\end{align}

Similar equations hold for each $j\in \Z$, assuming that $s\equiv 0$ whenever $j\le 0$ or $j> M$.
The resulting system of equations is easier to study if we move to the continuous version.

To this end,  we introduce a smooth function
 $u(x,t)$ on $[0,1]$,
for $(x,t)\in[-1,1]\times[0,\infty)$ such that $s(j,t) = \int_{I_j}u(x,t)dx$. This means that
$u$ restricted to the interval $I_j$ behaves like $s(j,t)/h$.

 Let us observe that
\begin{align*}
[ s(j+1)-s(j)]\sum_{i=1}^k s(j+1-i) = & \frac{h}{h} [ s(j+1)-s(j)]\sum_{i=1}^k s(j+1-i)\\
=& h \left[ \frac{s(j+1)-s(j)}{h}\right]\sum_{i=1}^k s(j+1-i)  \\
  \approx  &
h \frac{\partial u}{\partial x}\int_{0}^d u(x+y,t)dy,
\end{align*}
and similar formulas hold for the other differences and sums, so for $\tau$ and $h$ small, the
 equation of the continuous model reads:
\begin{align}
\frac{\tau N}{2h}\frac{\partial u(x,t)}{dt} = &
\frac{\partial u}{\partial x}\int_{0}^d [u(x-y,t) -  u(x+y,t)]dy   -d^2 u(x,t) \frac{\partial^2 u}{\partial x^2}
\label{EqCont}
\end{align}

If we assume that $d$ is small, only local interactions are allowed and we can get rid of the integral by doing a second order Taylor approximation,
 i.e.,
 $$u(x+y)=u(x)+y \frac{\partial u(x)}{\partial x}+\frac{y^2}{2} \frac{\partial^2 u(x)}{\partial x^2},$$
 and hence  Eq. (\ref{EqCont}) becomes:
\begin{align} \frac{\tau N}{2h}\frac{\partial u(x,t)}{dt} = & - \frac{\partial u}{\partial x}\int_{0}^d  2y \frac{\partial u}{\partial x}dy -d^2 u  \frac{\partial^2 u}{\partial x^2} \\
= &
-d^2\Big(\frac{\partial u}{\partial x}\Big)^2  -d^2 u  \frac{\partial^2 u}{\partial x^2} \\
= & -d^2 \frac{\partial^2}{\partial x^2}\Big(\frac{u^2}{2}\Big).
\label{EqCont_BC}
\end{align}
This equation must be complemented with some initial distribution of opinions at $t=0$, say
$$u(x,0)= u_0(x).$$

This equation is the well known Porous Media equation, reversed in time, similar to the one
obtained in \cite{WD2000} for $\mu=0.5$. The only difference is the constants in the time rescaling, $\tau
= 2h d^2/N$ here, since  $d = k h$ gives $\tau = O(d^3)$ as   in \cite{WD2000}, and the explicit constant in the D-W model depends
on $\mu$.

In Fig.\ref{Fig1} we compare the dynamics of this model with the one of D-W model, for $N=10000$ agents, $h=0.01$, two values
of bounded confidence parameter $d=1$ and $d=0.5$ and $\mu=0.5$ (see \cite{WD2000} for set up details).
We can observe that, for this set of parameters, the dynamics of convergence of our model is slower than in the  D-W model.

\begin{figure}
\includegraphics[width=1.0\textwidth]{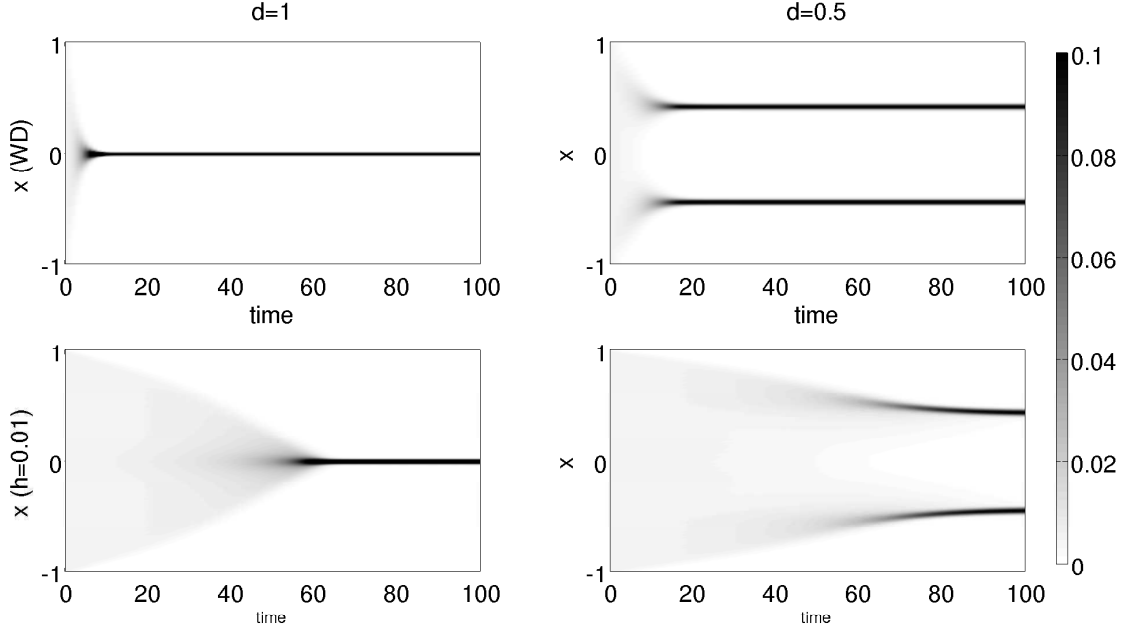}
\caption{ {\bf Comparative Dynamics with Bounded Confidence}. Time evolution of the D-W model, $\mu=0.5$ (Upper panels). The resulting dynamics of the system governed by Eq.\ref{EqDynMod1P}, for   $h=0.01$(Down panels). In both cases, simulations are done for a population of $N=10000$ agents and two values of $d\in\{1,0.5\}$.}
\label{Fig1}
\end{figure}

\section{Dynamics without Bounded Confidence}

The bounded confidence hypothesis, beside its sociological plausibility, is the key  to find a differential equation which describes the dynamics of the system \cite{WD2000}. However, the question that still remains to be answered is if it is possible to find an analytical solution to the equation of  the same model in the case where no bounded confidence hypothesis is invoked.  In this part of the paper we will address this issue.

If the bounded confidence hypothesis  is removed, every agent can interact with any other, independently of their differences in opinions. If we describe the dynamics in terms of density of agents with a given opinion, this means that the master equations turn to be non-local equations, the sums of Eq. (\ref{EqDiscr1}) are not truncated any more, and no longer Taylor approximation could be done.

Yet let start from the discrete equations again. In this case, they read:

\begin{align}
\frac{\tau N}{2}\frac{\partial s(j,t)}{\partial t} = &
 s(j+1,t) \sum_{i\le j} s(i,t)   +  s(j-1,t)  \sum_{i\ge j} s(i,t)  - s(j, t) \sum_{i\neq j} s(i,t) \nonumber  \\
= &
\Big(
 s(j+1,t) - s(j, t) \Big)\sum_{i\le j} s(i,t) - \nonumber  \\ &  -    \Big( s(j,t) - s(j-1,t) \Big) \sum_{i\ge j} s(i,t)  +2  s^2(j, t),
 \label{EqDiscr_sBC}
 \end{align}
for $1\le j \le M$.  The sum is in $\Z$ and, again, $s\equiv 0$ outside $[1,M]$. If we move these
equations to the continuous version in the same way as  before we obtain:
\begin{align*}
\frac{\tau N}{2h} \frac{\partial u(x,t)}{\partial t} = &
 \frac{\partial u(x,t)}{\partial x} \Big(
\int_{-\infty}^x u(y,t)dy     - \int_x^{\infty}  u(y,t)dy   \Big) +2  u^2(x, t) \nonumber \\ = &
\frac{\partial}{\partial x} \left(  u(x,t) \Big(
\int_{-\infty}^x u(y,t)dy     - \int_x^{\infty}  u(y,t)dy   \Big)\right),
 \end{align*}
 complemented with some initial distribution of opinions at $t=0$, say
$$u(x,0)= u_0(x).$$

We can see that this equation generalises the previous porous media equation, and it  is a
nonlocal first-order partial differential equation. From the mathematical point of view, this is
a nonlocal Porous Media Equation different from the ones obtained recently, see
\cite{CarVaz}, and it shares properties like mass preservation and finite time
propagation.

Finally, we re-scale times to get rid off the term $\tau N/2h$. We call $s=2ht/ N\tau$, and then
\begin{equation}
\frac{\partial u}{\partial s}=\frac{\partial u}{\partial t} \frac{dt}{ds}= \frac{\tau N}{2}\frac{\partial u}{\partial t},
\label{Tresc}
 \end{equation}
and then (renaming $s$ as $t$) we can consider the following equation
\begin{equation}
  \frac{\partial u(x,t)}{\partial t} =
\frac{\partial}{\partial x} \left(  u(x,t) \Big(
\int_{-\infty}^x u(y,t)dy     - \int_x^{\infty}  u(y,t)dy   \Big)\right).
 \label{EqCont_sBC}
 \end{equation}

\bigskip

Let us see if we can find a solution for this equation.

\subsection{Existence of solutions}

By a classical solution we understand $u\in C^{1,1}(\R\times(0,\infty))$ satisfying the
differential equation and
$$\lim_{(x,t)\to (x_0,0)} u(x,t)=u_0(x).$$
However, let us observe that we expect a measure $u$ as a solution, not necessarily
a differentiable function. So, we need to introduce a notion of a weak solution.

\begin{definition}
Given the following equation,
$$
\frac{\partial u(x,t)}{\partial t} = \frac{\partial }{\partial x} [u(x,t)G(t,x,F(x,t))],
$$
with initial condition $u(x,0)=u_0(x)$, where $F(x)= \int_{-\infty}^x u(y,t)dy$ is the cumulative
distribution function associated to the density $u$. We say that $u\in C^1((0,\infty),L^1)$ is a weak
solution if $u(x,0)=u_0(x)$, and
\begin{equation}\label{condition}
\frac{d}{dt} \int_{-\infty}^\infty h(x)u(x,t) dx =  -\int_{-\infty}^\infty h'(x) u(x,t)G(t,x,F(x,t)) dx
\end{equation}
for any $h\in C_0^1(\R)$.
\end{definition}

Condition  $u\in C^1((0,T),L^1(\R))$ means that, for each $t \in (0,T)$, the function $u( \cdot,
t)$ is an integrable function on $\R$, and this assignation is $C^1$ in the variable $t$.
Observe that weak solutions are not necessarily differentiable in the classic sense in the
variable $x$.

\bigskip
In what follows we are going to solve exactly Eq. (\ref{EqCont_sBC}) following  a method introduced in \cite{Toscani} in order to
deal with granular flows.
Giving $F(x)= \int_{-\infty}^x u(y,t)dy$, we will show that  we can re-write Eq. (\ref{EqCont_sBC}) as:
\begin{equation}
\frac{\partial u(x,t)}{\partial t} =  \frac{\partial}{\partial x} \Big(  u(x,t) [
2F(x)-1 ]  \Big).
 \label{EqContF_sBC}
\end{equation}

To this end, the method starts assigning a new variable for the cumulative function,
\begin{equation}
\rho = F(x) = \int_{-\infty}^x u(y,t)dy.
\end{equation}
When $u>0$ for any $x$ and $t$, we can introduce the inverse function $\X(\rho, t) =
F^{-1}(\rho,t)$. In other terms,
$$ \rho =
\int_{-\infty}^\X u(y,t)dy.$$
 However, since $u$ can be zero in some interval, it is convenient to  define
$$ \X(\rho,t) = \inf\{ x : F(x,t)\ge \rho\}.$$

With this change of variables, Eq. (\ref{EqContF_sBC}) becomes an infinite system of ordinary
differential equations,
\begin{equation}
 \frac{d \X(\rho, t)}{d t}  = 1-2\rho,
 \label{EqRho}
\end{equation}
one for each value of $\rho$, which can be solved explicitly as:
\begin{equation}
\X(\rho, t) = (1-2\rho)t + \X(\rho, 0),
\end{equation}
where $\X(\rho,0)$ is obtained from the initial datum, that is,
$$
\rho = \int_{-\infty}^{\X(\rho,0)} u_0(y)dy.
$$

We have obtained an implicit function for $\rho$, and  for each value of $t$, we can obtain it
in terms of $\X$, and since $\rho=F(\X)$, we recover the solution $u$ as
 $$
 u(x,t) = \partial_{\X} \rho(\X,t).
 $$

\bigskip
Let us prove the previous claims.  Let  us start with the following Lemma:

\begin{lema}\label{teoruno}
Let $h\in C^1_0(\R)$. Then
$$
\int_{-\infty}^\infty h(x) u(x,t)dx = \int_0^1 h(\X(\rho,t))d\rho.
$$
\end{lema}

\begin{proof}
Just change $x=\X$, and formally
$$ dx = \frac{\partial \X}{\partial \rho} d\rho =  \frac{\partial F^{-1}}{\partial \rho} d\rho= \frac{1}{u(x,t)}d\rho.$$
Hence,
$$
\int_{-\infty}^\infty h(x) u(x,t)dx = \int_0^1 h(\X(\rho,t))d\rho,
$$
and the proof is finished.
\end{proof}

We are ready to prove the main result:

\begin{theorem} Let $\X(\rho, t)$ be a solution of
\begin{equation}\label{eqteor2}\left\{\begin{array}{rclll}
\partial_t \X(\rho, t) & = & 1-2\rho, &\quad &  t\in(0,T),
\\
\X(\rho, 0) &= & \inf\left\{ x : \int_{-\infty}^x u_0(y)dy \ge \rho\right\}.
\end{array}\right.
\end{equation}
Then  there exists a weak solution of
 $$\left\{ \begin{array}{rclll}
\frac{\partial}{ \partial t} u(x,t) & = & \frac{\partial}{ \partial x} \left[u(x,t) (2F(x) - 1) \right]
 &  \quad & (x,t) \in \R\times (0,T) \\
 u(x,0) &=& u_0(x) & & x\in \R
\end{array}\right.
 $$
 with $\int_\R u(x,t)dx = 1$ for  $0<t<T$.
\end{theorem}

\begin{proof}
Take any function $h\in C^1_0(\R)$, and then, using Lemma \ref{teoruno}
\begin{align*}
\frac{d}{dt} \int_{-\infty}^\infty h(x)u(x,t) dx =& \frac{d}{dt} \int_0^1 h(\X(\rho,t)) d\rho \\
=& \int_0^1 h'(\X(\rho,t)) \frac{d \X(\rho, t)}{d t} d\rho.
\end{align*}
On the other hand,
$$ -\int_{-\infty}^\infty h'(x) u(x,t)(2F(x) - 1) dx=
- \int_0^1 h'(\X(\rho,t)) (2F(x) - 1)d\rho. $$

We get that condition (\ref{condition}) is satisfied if
$$\int_0^1 h'(\X(\rho,t)) \frac{d \X(\rho, t)}{d t} d\rho =
- \int_0^1 h'(\X(\rho,t)) (2\rho- 1) d\rho, $$
 or, equivalently,
$$\int_0^1 h'(\X(\rho,t))\Big[ \frac{d \X(\rho, t)}{d t} + 2\rho - 1\Big] d\rho = 0, $$
which trivially holds   if $\X(\rho, t)$ is a solution of Eq. (\ref{eqteor2}).

The proof is finished.
\end{proof}

With an extra effort, it could be proved also that the existence of a weak solution $u(x,t)$
implies the existence of solution $\X(\rho, t)$. Here, the simplicity of Eq. (\ref{eqteor2})
makes unnecessary such equivalence.

\bigskip

\subsection{Convergence to the mean}

Observe that
$$\partial_t \X(\rho, t)=1-2\rho$$ is positive   for $\rho<1/2$, and negative for $\rho>1/2$.
Therefore, since $\rho=1/2$ gives  the median of the distribution, we get that
$$
    \X(\rho, t) = \inf\left\{ x : \int_{-\infty}^x v(y,t)dy =\rho\right\}
  $$
 strictly  increases for $0<\rho<1/2$, and decreases for $1/2<\rho<1$. Hence, there exists some $c_0$ such that
$$
\lim_{t\to T} \X(\rho, t) = c_0,
 $$
which implies that $u(x,t)\to \delta_{c_0}$ as $t\to T$. In other words, the distribution of
opinions concentrates on the mean value of the distribution and the population  reaches consensus.

\subsection{Examples with different initial conditions}

Let's solve explicitly the theoretical asymptotic values of $u(x,t)$ for two different initial conditions: a symmetrical one ($u(x,0)=cte$) and an asymmetrical ($u(x,0)=2x$). We  compare them with computer simulations in both cases. We will see that computer simulations give a strong agreement with these specific examples.

\subsubsection{a.- Constant opinion's initial distribution}

Let $u_0(x)= \tfrac{1}{2} \chi_{[-1,1]}(x)$, where $\chi_{[a,b]}(x)$ is equal to one if $x\in [a,b]$ and zero outside. Then,
$$
\X(\rho, 0) = \inf\left\{ x : \int_{-1}^x dy  \ge \rho\right\}=2\rho-1.
$$

So,
$$
\X(\rho, t) = (1-2\rho)t + \X(\rho, 0) = (1-2 \rho)t+2\rho-1,
$$
 and inverting, since $0\le \rho\le 1$,

$$
\rho =\left( \frac{x-t+1}{2-2t}\right) \chi_{[t-1,1-t]}(x)
$$

Finally, for $0\le t<1$, since $u\ge 0$,
$$
u(x,t)= \partial_x \left( \frac{x-t+1}{2-2t}\chi_{[t-1,1-t]}(x)\right) = (2-2t)^{-1}\chi_{[t-1,1-t]}(x).
$$

Observe that the solution blows up when $t$ reaches $1$,
and
$$ \lim_{t\to 1} u(x,t) = \delta_{0}.$$

and
$$
\int_{-1}^1 x u(x,t)dx=0,
$$

\subsubsection{b.- Linear opinion's initial distribution}

  Let $u(x,0)= \frac{x+1}{2} \chi_{[-1,1]}(x)$. Then,

\begin{align*}
F(x,0) = & \int_{-\infty}^x \frac{y+1}{2} \chi_{[-1,1]}(y) dy \\
= & \int_{-1}^x  \frac{y+1}{2} \chi_{[-1,1]}(y) dy \\
 = & \frac{(x+1)^2}{4}\chi_{[-1,1]}(x).
\end{align*}

 Hence, $ \X(\rho,t)= F^{-1}(\rho)$  and we have
 $$
 \X(\rho,0)=(2\sqrt{\rho}-1)\chi_{[0,1]}(\rho).
 $$

 A direct computation gives
$$
\X(\rho,t)=(1-2\rho)t+(2\sqrt{\rho}-1), \qquad 0\le \rho \le 1,
$$
and we recover $u(x,t)=\partial_x \X^{-1}(x)$ from
$$ x=(1-2\rho)t+2\sqrt{\rho}-1, \qquad 0\le \rho \le 1,
$$
by computing the inverse function as before.

Finally, observe that, for any $t$,
$$
\int_{-1}^1 x u(x,t)dx=
\int_0^1 \X(\rho,t)d\rho=\frac{1}{3},
$$
and so $u(x,t)\to \delta_{1/3}$.


In Fig.\ref{Fig2} we can observe the dynamics of $u(x,t)$ obtained by simulating $N=10000$ agents, for $h=0.01$ for the two mentioned initial conditions $u(x,0)$ (uniform and linear), and in Fig.(\ref{Fig3}) the dynamics of the medians. 
We have to take into account that the relation between the re-scaled time used to derive the analytical solution and the time used in simulations differ by a factor $(2h)/(N\tau)$, 
with $\tau=1/N$. This factor makes that the theoretical re-scaled time corresponds to the time used in simulations divided by $50$.
We can observe how the distributions and the median converge to the predicted theoretical  values $0$  and $1/3$, at  times predicted in theory ($t=50$, rescaled is $T=1$) showing the perfect agreement between theory and simulations.

\begin{figure}
\includegraphics[width=0.8\textwidth]{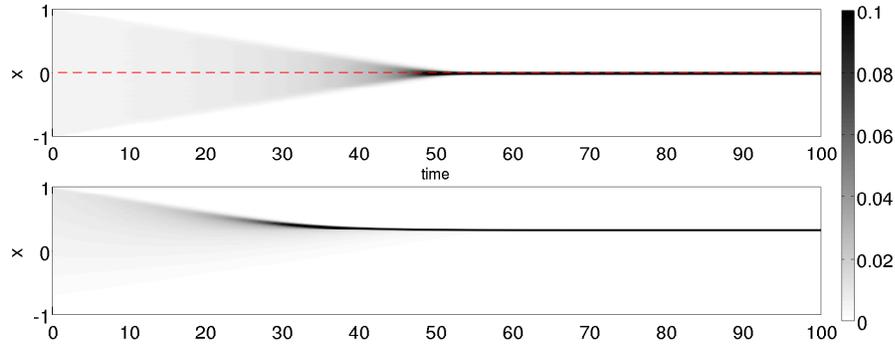}
\caption{ {\bf Comparative Dynamics without Bounded Confidence}. Time evolution of an agent's based model governed by Eq.\ref{EqDynMod1P}, for a population of $N=10000$ agents and $h=0.01$, and two initial distributions of opinions: Uniform (Upper panel) and Linear (Down panel).}
\label{Fig2}
\end{figure}

\begin{figure}
\includegraphics[width=0.8\textwidth]{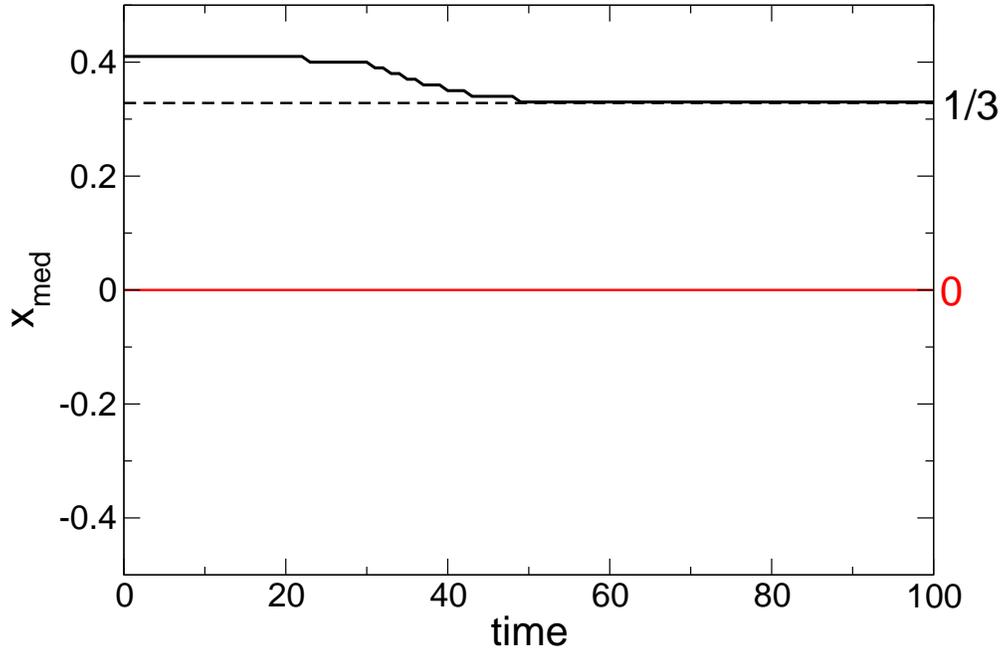}
\caption{ {\bf Median and average value dynamics with Bounded Confidence}. Time evolution of the median and the average value of the distribution  for a population of $N=10000$ agents and $h=0.01$, and two initial distributions of opinions: Uniform (Upper panel) and Linear (Down panel).}
\label{Fig3}
\end{figure}

 \section{Conclusions}

We  presented a model of opinion formation, based on the bounded confidence model  (D-W) presented in \cite{WD2000}, 
 in which  consensus emerges as a  results of a collective discussion process.
We derived analytical solutions to a non-local master equation both for the case of bounded
confidence dynamics, as it was done in the original D-W model, as  for the case of long
range interactions which has not been done before. 

In the first case, for a small bounded
confidence, the distribution of opinions is ruled by a Porous Media equation, reversed in
time, similar to the one obtained in \cite{WD2000}, with the only difference in the time rescaling.

 In the second case, without bounded confidence, the distribution of opinions converges to Dirac's delta function
concentrated at the mean value of the distribution, and the population always reaches
consensus. The simulations performed with an agent based model show an excellent agreement with theoretical results.


\begin{thebibliography}{14}%
\makeatletter
\providecommand \@ifxundefined [1]{%
 \@ifx{#1\undefined}
}%
\providecommand \@ifnum [1]{%
 \ifnum #1\expandafter \@firstoftwo
 \else \expandafter \@secondoftwo
 \fi
}%
\providecommand \@ifx [1]{%
 \ifx #1\expandafter \@firstoftwo
 \else \expandafter \@secondoftwo
 \fi
}%
\providecommand \natexlab [1]{#1}%
\providecommand \enquote  [1]{``#1''}%
\providecommand \bibnamefont  [1]{#1}%
\providecommand \bibfnamefont [1]{#1}%
\providecommand \citenamefont [1]{#1}%
\providecommand \href@noop [0]{\@secondoftwo}%
\providecommand \href [0]{\begingroup \@sanitize@url \@href}%
\providecommand \@href[1]{\@@startlink{#1}\@@href}%
\providecommand \@@href[1]{\endgroup#1\@@endlink}%
\providecommand \@sanitize@url [0]{\catcode `\\12\catcode `\$12\catcode
  `\&12\catcode `\#12\catcode `\^12\catcode `\_12\catcode `\%12\relax}%
\providecommand \@@startlink[1]{}%
\providecommand \@@endlink[0]{}%
\providecommand \url  [0]{\begingroup\@sanitize@url \@url }%
\providecommand \@url [1]{\endgroup\@href {#1}{\urlprefix }}%
\providecommand \urlprefix  [0]{URL }%
\providecommand \Eprint [0]{\href }%
\providecommand \doibase [0]{http://dx.doi.org/}%
\providecommand \selectlanguage [0]{\@gobble}%
\providecommand \bibinfo  [0]{\@secondoftwo}%
\providecommand \bibfield  [0]{\@secondoftwo}%
\providecommand \translation [1]{[#1]}%
\providecommand \BibitemOpen [0]{}%
\providecommand \bibitemStop [0]{}%
\providecommand \bibitemNoStop [0]{.\EOS\space}%
\providecommand \EOS [0]{\spacefactor3000\relax}%
\providecommand \BibitemShut  [1]{\csname bibitem#1\endcsname}%
\let\auto@bib@innerbib\@empty
\bibitem [{\citenamefont {Clifford}\ and\ \citenamefont
  {Susbury}(1973)}]{Voter1}%
  \BibitemOpen
  \bibfield  {author} {\bibinfo {author} {\bibfnamefont {P.}~\bibnamefont
  {Clifford}}\ and\ \bibinfo {author} {\bibfnamefont {A.}~\bibnamefont
  {Susbury}},\ }\href {\doibase 10.1093/biomet/60.3.581} {\bibfield  {journal}
  {\bibinfo  {journal} {Biometrika}\ }\textbf {\bibinfo {volume} {60}},\
  \bibinfo {pages} {581} (\bibinfo {year} {1973})},\ \Eprint
  {http://arxiv.org/abs/http://biomet.oxfordjournals.org/content/60/3/581.full.pdf+html}
  {http://biomet.oxfordjournals.org/content/60/3/581.full.pdf+html}
  \BibitemShut {NoStop}%
\bibitem [{\citenamefont {Liggett}(1997)}]{Voter2}%
  \BibitemOpen
  \bibfield  {author} {\bibinfo {author} {\bibfnamefont {T.~M.}\ \bibnamefont
  {Liggett}},\ }\href {\doibase 10.1214/aop/1024404276} {\bibfield  {journal}
  {\bibinfo  {journal} {Ann. Probab.}\ }\textbf {\bibinfo {volume} {25}},\
  \bibinfo {pages} {1} (\bibinfo {year} {1997})}\BibitemShut {NoStop}%
\bibitem [{\citenamefont {Deffuant}\ \emph {et~al.}(2000)\citenamefont
  {Deffuant}, \citenamefont {Neau}, \citenamefont {Amblard},\ and\
  \citenamefont {Weisbuch}}]{WD2000}%
  \BibitemOpen
  \bibfield  {author} {\bibinfo {author} {\bibfnamefont {G.}~\bibnamefont
  {Deffuant}}, \bibinfo {author} {\bibfnamefont {D.}~\bibnamefont {Neau}},
  \bibinfo {author} {\bibfnamefont {F.}~\bibnamefont {Amblard}}, \ and\
  \bibinfo {author} {\bibfnamefont {G.}~\bibnamefont {Weisbuch}},\ }\href
  {\doibase 10.1142/S0219525900000078} {\bibfield  {journal} {\bibinfo
  {journal} {Adv. Complex Syst.}\ }\textbf {\bibinfo {volume} {3}},\ \bibinfo
  {pages} {87} (\bibinfo {year} {2000})}\BibitemShut {NoStop}%
\bibitem [{\citenamefont {M\"{a}s}\ and\ \citenamefont
  {Flache}(2013)}]{Masetal2013}%
  \BibitemOpen
  \bibfield  {author} {\bibinfo {author} {\bibfnamefont {M.}~\bibnamefont
  {M\"{a}s}}\ and\ \bibinfo {author} {\bibfnamefont {A.}~\bibnamefont
  {Flache}},\ }\href {\doibase 10.1371/journal.pone.0074516} {\bibfield
  {journal} {\bibinfo  {journal} {PLoS ONE}\ }\textbf {\bibinfo {volume} {8}},\
  \bibinfo {pages} {e74516} (\bibinfo {year} {2013})}\BibitemShut {NoStop}%
\bibitem [{\citenamefont {La~Rocca}\ \emph {et~al.}(2014)\citenamefont
  {La~Rocca}, \citenamefont {Braunstein},\ and\ \citenamefont
  {Vazquez}}]{LaRoccaetal2014}%
  \BibitemOpen
  \bibfield  {author} {\bibinfo {author} {\bibfnamefont {C.~E.}\ \bibnamefont
  {La~Rocca}}, \bibinfo {author} {\bibfnamefont {L.~A.}\ \bibnamefont
  {Braunstein}}, \ and\ \bibinfo {author} {\bibfnamefont {F.}~\bibnamefont
  {Vazquez}},\ }\href {\doibase doi:10.1209/0295-5075/106/40004} {\bibfield
  {journal} {\bibinfo  {journal} {EPL (Europhysics Letters)}\ }\textbf
  {\bibinfo {volume} {106}},\ \bibinfo {pages} {40004} (\bibinfo {year}
  {2014})}\BibitemShut {NoStop}%
\bibitem [{\citenamefont {Thorndike}(1938)}]{Thorndike}%
  \BibitemOpen
  \bibfield  {author} {\bibinfo {author} {\bibfnamefont {R.~L.}\ \bibnamefont
  {Thorndike}},\ }\href {\doibase 10.1080/00224545.1938.9920036} {\bibfield
  {journal} {\bibinfo  {journal} {The Journal of Social Psychology}\ }\textbf
  {\bibinfo {volume} {9}},\ \bibinfo {pages} {343} (\bibinfo {year} {1938})},\
  \Eprint
  {http://arxiv.org/abs/http://dx.doi.org/10.1080/00224545.1938.9920036}
  {http://dx.doi.org/10.1080/00224545.1938.9920036} \BibitemShut {NoStop}%
\bibitem [{\citenamefont {Nordhoy}(1962)}]{Nordhoy}%
  \BibitemOpen
  \bibfield  {author} {\bibinfo {author} {\bibfnamefont {F.}~\bibnamefont
  {Nordhoy}},\ }\emph {\bibinfo {title} {Group Interaction in Decision-Making
  Under Risk}},\ \href@noop {} {Master's thesis},\ \bibinfo  {school} {School
  of Industrial Management, Massachusetts Insti- tute of Technology.} (\bibinfo
  {year} {1962})\BibitemShut {NoStop}%
\bibitem [{\citenamefont {Jean}(1970)}]{Jean}%
  \BibitemOpen
  \bibfield  {author} {\bibinfo {author} {\bibfnamefont {R.~S.}\ \bibnamefont
  {Jean}},\ }\href@noop {} {\bibfield  {journal} {\bibinfo  {journal}
  {Procedings of the 78th Anual Convention of the American Psychological
  Asociation}\ }\textbf {\bibinfo {volume} {5}},\ \bibinfo {pages} {339}
  (\bibinfo {year} {1970})}\BibitemShut {NoStop}%
\bibitem [{\citenamefont {Stoner}(1968)}]{Stoner}%
  \BibitemOpen
  \bibfield  {author} {\bibinfo {author} {\bibfnamefont {J.}~\bibnamefont
  {Stoner}},\ }\href@noop {} {\bibfield  {journal} {\bibinfo  {journal} {J.
  Exp. Soc. Psychol, MIT}\ }\textbf {\bibinfo {volume} {4}},\ \bibinfo {pages}
  {442} (\bibinfo {year} {1968})}\BibitemShut {NoStop}%
\bibitem [{\citenamefont {Vinokur}(1971)}]{Vinokur}%
  \BibitemOpen
  \bibfield  {author} {\bibinfo {author} {\bibfnamefont {A.}~\bibnamefont
  {Vinokur}},\ }\href {\doibase http://dx.doi.org/10.1037/h0031926} {\bibfield
  {journal} {\bibinfo  {journal} {Journal of Personality and Social
  Psychology}\ }\textbf {\bibinfo {volume} {20}},\ \bibinfo {pages} {472}
  (\bibinfo {year} {1971})}\BibitemShut {NoStop}%
\bibitem [{\citenamefont {Burnstein}\ and\ \citenamefont
  {Vinokur}(1973)}]{Burnstein}%
  \BibitemOpen
  \bibfield  {author} {\bibinfo {author} {\bibfnamefont {E.}~\bibnamefont
  {Burnstein}}\ and\ \bibinfo {author} {\bibfnamefont {A.}~\bibnamefont
  {Vinokur}},\ }\href {\doibase http://dx.doi.org/10.1016/0022-1031(73)90004-8}
  {\bibfield  {journal} {\bibinfo  {journal} {Journal of Experimental Social
  Psychology}\ }\textbf {\bibinfo {volume} {9}},\ \bibinfo {pages} {123 }
  (\bibinfo {year} {1973})}\BibitemShut {NoStop}%
\bibitem [{\citenamefont {Levine}\ and\ \citenamefont
  {Payne}(1974)}]{Levine74}%
  \BibitemOpen
  \bibfield  {author} {\bibinfo {author} {\bibfnamefont {H.}~\bibnamefont
  {Levine}}\ and\ \bibinfo {author} {\bibfnamefont {L.~E.}\ \bibnamefont
  {Payne}},\ }\href {\doibase -} {\bibfield  {journal} {\bibinfo  {journal}
  {Journal of Differential Equations}\ }\textbf {\bibinfo {volume} {16}},\
  \bibinfo {pages} {319–334} (\bibinfo {year} {1974})}\BibitemShut {NoStop}%
\bibitem [{\citenamefont {Li}\ and\ \citenamefont {Toscani}(2004)}]{Toscani}%
  \BibitemOpen
  \bibfield  {author} {\bibinfo {author} {\bibfnamefont {H.}~\bibnamefont
  {Li}}\ and\ \bibinfo {author} {\bibfnamefont {G.}~\bibnamefont {Toscani}},\
  }\href {\doibase 10.1007/s00205-004-0307-8} {\bibfield  {journal} {\bibinfo
  {journal} {Archive for Rational Mechanics and Analysis}\ }\textbf {\bibinfo
  {volume} {172}},\ \bibinfo {pages} {407} (\bibinfo {year}
  {2004})}\BibitemShut {NoStop}%
\bibitem [{\citenamefont {Carrillo}\ and\ \citenamefont
  {V{\'a}zquez}(2015)}]{CarVaz}%
  \BibitemOpen
  \bibfield  {author} {\bibinfo {author} {\bibfnamefont {J.~A.}\ \bibnamefont
  {Carrillo}}\ and\ \bibinfo {author} {\bibfnamefont {J.~L.}\ \bibnamefont
  {V{\'a}zquez}},\ }\href {\doibase 10.1098/rsta.2014.0275} {\bibfield
  {journal} {\bibinfo  {journal} {Philosophical Transactions of the Royal
  Society of London A: Mathematical, Physical and Engineering Sciences}\
  }\textbf {\bibinfo {volume} {373}} (\bibinfo {year} {2015}),\
  10.1098/rsta.2014.0275},\ \Eprint
  {http://arxiv.org/abs/http://rsta.royalsocietypublishing.org/content/373/2050/20140275.full.pdf}
  {http://rsta.royalsocietypublishing.org/content/373/2050/20140275.full.pdf}
  \BibitemShut {NoStop}%
\end{thebibliography}
\end{document}